\documentclass[doublespacing]{siamart171218}

\usepackage{setspace}

	\usepackage{graphicx}

\makeatletter
\def\maxwidth{
  \ifdim\Gin@nat@width>\linewidth
    \linewidth
  \else
    \Gin@nat@width
  \fi
}
\makeatother

\usepackage{placeins}
\usepackage{amsmath}
\usepackage{amsfonts}
\usepackage{amssymb}
\usepackage{multirow}
\usepackage[noadjust]{cite}
\usepackage{fixltx2e}
\usepackage{mathtools}
\usepackage{blindtext}
\usepackage{verbatim}
\newtheorem{definitions}{Definition}

\newcommand{\Mp}{\mathcal{P}}
\newcommand{\Bmp}{\overline{\mathcal{P}}}
\newcommand{\tkr}{\text{k-r}}
\newcommand{\appropto}{\mathrel{\vcenter{
  \offinterlineskip\halign{\hfil$##$\cr
    \propto\cr\noalign{\kern2pt}\sim\cr\noalign{\kern-2pt}}}}}

\begin{document}
\title{Effective resistance preserving directed graph symmetrization \thanks{\funding{This work was supported by the Alexander von Humboldt Foundation with funds from the German Federal Ministry of Education and Research (BMBF).} }}
\author{Katherine Fitch \thanks{Katherine Fitch is with the Chair of Operations Research, Technical University of Munich, Munich, Germany. katie.fitch@tum.de} }

\maketitle

\begin{abstract}
This work presents a new method for symmetrization of directed graphs that constructs an undirected graph with equivalent pairwise effective resistances as a given directed graph. Consequently a graph metric, square root of effective resistance, is preserved between the directed graph and its symmetrized version. It is shown that the preservation of this metric allows for interpretation of algebraic and spectral properties of the symmetrized graph in the context of the directed graph, due to the relationship between effective resistance and the Laplacian spectrum. Additionally, Lyapunov theory is used to demonstrate that the Laplacian matrix of a directed graph can be decomposed into the product of a projection matrix, a skew symmetric matrix, and the Laplacian matrix of the symmetrized graph. The application of effective resistance preserving graph symmetrization is discussed in the context of spectral graph partitioning and Kron reduction of directed graphs. 
\end{abstract}
\section{Introduction}  
In general, a large number of graph-theoretic problems and applications such as clustering \cite{schaeffer2007}, signal processing over graphs \cite{shuman2013}, and edge sparsification \cite{spielman2011, koutis2016} are well-studied and understood in the context of undirected graphs, but less so in the context of directed graphs. The solutions to these problems rely on algebraic and spectral graph theory that has been derived for undirected graphs but cannot be immediately applied to directed graphs due to the loss of symmetry in the Laplacian matrix.  A natural solution approach is to symmetrize the Laplacian and solve the problem on the associated symmetrized graph \cite{Malliaros2013, satuluri2011, wang2010}.

In \cite{Malliaros2013}, the authors provide a comprehensive survey of common symmetrization methods used for clustering directed graphs including bibliometric symmetrization \cite{satuluri2011}, degree-discounted symmetrization \cite{he2014}, or simply ignoring edge direction. These methods have the drawback that there is not always a formal guarantee that the symmetrized Laplacian captures appropriate or desired characteristics, such as distance between nodes, of the original directed Laplacian. Random walk-based symmetrization has been used to define the Cheeger inequality for directed graphs \cite{chung2005},   spectral clustering of directed graphs \cite{zhou2005}, and other applications. However, the definitions in \cite{chung2005} hold only when the graph is strongly connected, or in other words, for every pair of nodes $i$, $j$ there exists a directed path from $i$ to $j$.

One property of undirected graphs that has recently been extended to directed graphs is the notion of effective resistance \cite{young20161, young20162}. In undirected graphs, effective resistance is commonly known as resistance distiance \cite{klein1993} and arises from considering a graph as a network of resistors where each edge is replaced by a resistor with resistance equal to the inverse of the edge weight. The resistance distance between two nodes $i,j$ is simply the net resistance resulting from connecting a voltage source between $i$, $j$. As demonstrated in \cite{klein1993}, resistance distance can be equivalently computed from the pseudoinverse of the undirected graph Laplacian. Resistance distance in undirected graphs has been shown to be proportional to the expected length of a random walk between two nodes, also known as the commute time \cite{Chandra1996}, and is closely related to concepts in Markov chains and random walks  \cite{Tetali1991, Ghosh2008}. A drawback of resistance distance is that it becomes less informative of graph structure when the size of the graph is very large, though there are adaptations to resistance distance that seek to compensate for this deficiency \cite{Luxburg2010}. Effective resistance in directed graphs is a generalization of undirected resistance distance and its square root is a graph metric \cite{young20161, young20162}.

The work presented here builds off of \cite{young20161, young20162} and applies Lyapunov theory to define a new graph symmetrization method by calculating the undirected graph that has equivalent effective resistance to that of a given directed graph. In doing so, it follows immediately that the undirected, symmetrized graph preserves a metric of the directed graph, namely, the square root of effective resistance. Moreover, it is shown that the Laplacian matrix of a directed graph can be decomposed into the product of a projection matrix, a skew symmetric matrix, and the Laplacian matrix of the symmetrized graph. This mapping between a directed Laplacian and its symmetrized version, along with the preservation of a metric on graphs, allows certain tools from algebraic and spectral graph theory to be applied to a symmetrized graph with clear interpretation in the context of the original directed graph. Therefore, the proposed symmetrization method provides an improvement upon heuristic symmetrizations where such an interpretation is not possible. 

The paper is organized as follows. Section \ref{sec:not} provides an overview of relevant notation. In Section \ref{sec:effr}, the generalized notion of effective resistance from \cite{young20161, young20162} is reviewed. Section \ref{sec:sym} introduces effective resistance preserving graph symmetrization and the matrix relationship between a directed Laplacian and the symmetrized Laplacian. Relevant spectral properties of the symmetrized Laplacian are given in Section \ref{sec:spec}. Two relevant applications, graph bisection and node sparsification, are discussed in Sections \ref{sec:bisec} and \ref{sec:spars}, respectively. Section \ref{sec:large} briefly discusses analysis of large graphs. The paper is concluded with final remarks in Section~\ref{sec:fr}.

\section{Notation} \label{sec:not}
Let $\mathcal{G} = (\mathcal{V}, \mathcal{E}, A)$ be a connected, directed graph, where $\mathcal{V} = \{1,2,\dots, n\}$ is the set of vertices, and $\mathcal{E} \subseteq \mathcal{V} \times \mathcal{V}$ is the set of $m$ edges.  $A \in \mathbb{R}^{n \times n} $ is the adjacency matrix where element $a_{i,j}$ is the nonnegative weight on edge $(i,j)$. If $(i,j) \in \mathcal{E}$, then $a_{i,j} > 0 $; otherwise $a_{i,j} = 0$.  Graphs with self loops, i.e., with an edge from a node to itself, are not considered. The \emph{out-degree} of node $i$ is calculated as $d_i = \sum_{j=1}^n a_{i,j}$.  The out-degree matrix is a diagonal matrix of node out-degrees, $D = \text{diag}\{d_1, d_2, \dots,d_n\}$. The associated directed \emph{Laplacian} matrix is defined as $L=D-A$.  In constructing the directed Laplacian matrix the precedent is followed that out-degrees are on the main diagonal and row sums are equal to zero, that is, $L\mathbf{1}_n = 0$. Let the vector of column sums of $L$ be $\mathbf{1}_n^TL = \mathbf{\Delta}_d$.  The graphs considered are \emph{connected} in the sense that there exists at least one globally reachable node, $k$. In other words, there is a path from every node $i$ in $\mathcal{G}$ to $k$. As noted in \cite{young20161}, this notion of connectivity lies in between the classical definitions of strong and weak connectivity. The Moore--Penrose pseudoinverse of $L$ is $L^+$. 

A graph is said to be undirected if $A$ is symmetric, that is, $(i,j) \in \mathcal{E}$ implies $(j,i) \in \mathcal{E}$ and $a_{i,j} = a_{j,i}$. If an undirected graph has been calculated by the symmetrization method described in this paper, then edge weights $a_{i,j}$ are permitted to be negative. The Laplacian matrix of an undirected graph is also symmetric. An undirected graph is connected if and only if there is a path between any pair of nodes $(i,j)$. Sets and matrices corresponding to undirected graphs are indexed with $u$; for example, $L_u$ is the undirected Laplacian associated with undirected graph $\mathcal{G}_u$. 

The $i$th eigenvalue of a matrix $Z$ is denoted by $\lambda_i(Z)$, where the eigenvalues are ordered such that $\lambda_1(Z)\leq \dots \leq \lambda_i(Z) \leq \dots \leq \lambda_n(Z). $ The eigenvalues of a directed Laplacian matrix either are 0 or have positive real part \cite{Agaev2005}. The eigenvalues of an undirected Laplacian matrix are all real and nonnegative. If a graph (directed or undirected) is connected as defined above, then the associated Laplacian matrix will have one eigenvalue of 0.
 
Let $\mathbf{e}_n^{(k)}$ be the $k$th standard basis vector for $\mathbb{R}^n$. Let $P_n = I_n - \frac{1}{n}\mathbf{1}_n\mathbf{1}_n^T$. Let $\mathbf{1}^{\perp}_n =$span$\{\mathbf{1}_n\}^{\perp}$ be the subspace of $\mathbb{R}^n$ perpendicular to $\mathbf{1}_n$. Let $Q \in \mathbb{R}^{(n - 1) \times n}$ be a matrix with rows that form an orthonormal basis for $\mathbf{1}^{\perp}_n $. Then the following properties hold:
\begin{align}
Q\mathbf{1}_n = \mathbf{0}, \;\; QQ^T = I_{n-1}, \;\; Q^TQ=P_n. \label{Q}
\end{align}

The \emph{reduced Laplacian} matrix is defined as $\bar{L} = QLQ^T$ and characterizes the Laplacian matrix on $\mathbf{1}^{\perp}_n $. $\bar{L}$ has the same eigenvalues as $L$ except for a 0 eigenvalue and is therefore invertible if the graph is connected \cite{young2010}. 

Let $H \in \mathbb{R}^{n \times n}$ be a projection matrix onto $\mathbf{1}^{\perp}_n $, where $H$ is not necessarily an orthogonal projection matrix and $HL = L$, that is, the image of $L$ is contained in the kernel of $(H-I_n)$. 

 For a subset $\mathcal{P} \subset \mathcal{V}$, let $\overline{\mathcal{P}} \subset \mathcal{V}$ be the complementary subset. A directed cut between two subsets $(\mathcal{P},\overline{\mathcal{P}})$ is defined as the sum of weights of edges directed from ${\mathcal{P}}$ toward $\overline{\mathcal{P}}$. Equivalently,
 \begin{align}
 \text{cut}(\mathcal{P},\overline{\mathcal{P}}) = \frac{1}{2} y^T (L+L^T) y, \nonumber
 \end{align}
 where
 \begin{align}
 y_i = \begin{cases}
 1 & \text{if } i \in \mathcal{P}, \\
  0 & \text{if }  i \in \overline{\mathcal{P}}. \label{yvec}
\end{cases} 
\end{align}
For undirected graphs, $\text{cut}(\mathcal{P},\overline{\mathcal{P}})  = \text{cut}(\overline{\mathcal{P}},{\mathcal{P}}) $.

\section{Effective resistance} \label{sec:effr}
The notion of effective resistance generalizes resistance distance in undirected graphs such that directed graphs can also be considered. This section provides a brief review of effective resistance as defined in the two-part paper \cite{young20161, young20162}. 
\begin{definitions} \label{defres} \cite{young20161}: Let $\mathcal{G}$ be a connected, directed graph with $n$ nodes and Laplacian matrix $L$. Then the effective resistance between nodes $i$ and $j$ in $\mathcal{G}$ is defined as
\begin{align}
r_{i,j} &= \Big( \mathbf{e}_n^{(i)} - \mathbf{e}_n^{(j)} \Big)^T X \Big( \mathbf{e}_n^{(i)} - \mathbf{e}_n^{(j)} \Big) \nonumber \\
& = x_{i,i} + x_{j,j} - 2x_{i,j}, \label{res}
\end{align}
where
\begin{align}
X & = 2Q^T \Sigma Q, \label{getX}\\
\bar{L} \Sigma + \Sigma \bar{L}^T &= I_{n-1}, \label{lyapeff} \\
\bar{L} & = Q L Q^T, \label{lbar}
\end{align}
and $Q$ is a matrix satisfying (\ref{Q}).
\end{definitions}
The solution, $\Sigma$, to (\ref{lyapeff}) has a unique, symmetric, positive definite solution when all eigenvalues of $\bar{L}$ have positive real part \cite{Dullerud2000}. Therefore, $\Sigma$ is invertible whenever the graph associated with $L$ is connected. For an undirected graph with Laplacian $L_u$ and reduced Laplacian $\bar{L}_u = Q L_u Q^T$, it can be observed that the solution to (\ref{lyapeff}) is $\Sigma = \frac{1}{2} \bar{L}_u  ^{-1}$. Thus, $X = L_u^+$ and $r_{u_{i,j}}= ( \mathbf{e}_n^{(i)} - \mathbf{e}_n^{(j)} )^T L_u^+ ( \mathbf{e}_n^{(i)} - \mathbf{e}_n^{(j)} ) $. Therefore, Definition \ref{defres} is consistent with the classical notion of resistance distance in undirected graphs by Klein and Rand\'{i}c \cite{klein1993}.

The following three properties of effective resistance were proven in \cite{young20161}:
\begin{enumerate}
\item Effective resistance is well-defined.
\item Effective resistance depends on connections between nodes.
\item Effective resistance is a distance-like function and its
square root is a metric.
\end{enumerate}

\section{Undirected and directed Laplacians with equivalent effective resistance} \label{sec:sym}
 The following section demonstrates that for any connected, directed graph with Laplacian matrix $L$, there exists a symmetric, undirected Laplacian $\hat{L}_u$ on the same set of nodes and possibly admitting negative edge weights, for which the effective resistance between any pair of nodes $(i,j)$ in $L$ is equal to the effective resistance between $(i,j)$ in $\hat{L}_u$.
 
\begin{proposition}\label{thm1}
Let $\mathcal{G} = (\mathcal{V},\mathcal{E},A)$ be a directed, connected graph of order $n$. Then, there exists an undirected graph $\mathcal{\hat{G}}_u= (\mathcal{V},\mathcal{\hat{E}}_u,\hat{A}_u)$,  possibly admitting negative edge weights, for which effective resistance in $\mathcal{G}$ is equivalent to effective resistance in $\mathcal{\hat{G}}_u$ for all $i,j \in \mathcal{V}$.
\end{proposition}
\begin{proof}
It can be observed that the unique solution to (\ref{lyapeff}), $\Sigma$, also satisfies the following trivial Lyapunov equation with its inverse:
\begin{align}
\frac{1}{2}\Sigma^{-1}\Sigma + \frac{1}{2}\Sigma \Sigma^{-1} = I_{n-1}.  \label{siglyap}
\end{align}
Additionally, the pseudoinverse of $X$ can be expressed as $X^+ = (2Q^T \Sigma Q)^+ = \frac{1}{2} Q^T \Sigma^{-1} Q$. The matrix $X^+$, therefore, is an $n \times n$ symmetric, positive semidefinite matrix with zero row and column sums. Therefore, $X^+$ can be interpreted as an undirected Laplacian matrix, $\hat{L}_u = X^+$, where $\hat{L}_u$ potentially admits negative edge weights and is associated with an undirected graph, $\mathcal{\hat{G}}_u= (\mathcal{V},\mathcal{\hat{E}}_u,\hat{A}_u)$, on the same set of nodes as $\mathcal{G}$. Since $L$, the Laplacian associated with $\mathcal{G}$,  and $\hat{L}_u$ have the same solution to the Lyapunov equation, (\ref{lyapeff}), it follows immediately that effective resistance $\mathcal{G}$ is equivalent to effective resistance in $\mathcal{\hat{G}}_u$ for all $i,j \in \mathcal{V}$. 
\end{proof}
It follows from Proposition \ref{thm1} calculating an undirected Laplacian $\hat{L}_u$ with equivalent effective resistance to a directed Laplacian, $L$ can be accomplished by the following steps:
\begin{enumerate}
\item Calculate the reduced Laplacian, $\bar{L} = QLQ^T$.
\item Solve the Lyapunov equation $\bar{L} \Sigma + \Sigma \bar{L}^T = I_{n-1}$.
\item Project from $\mathbb{R}^{(n-1)\times (n-1)}$ to $\mathbb{R}^{n \times n}$ by $X  = 2Q^T \Sigma Q$.
\item Calculate the pseudoinverse $\hat{L}_u = X^+$.
\end{enumerate}
The resulting undirected Laplacian can be thought of as a symmetrization of a directed Laplacian where a metric on graphs, square root of effective resistance, has been preserved. This is in contrast with heuristic symmetrization methods where there is no guarantee that the symmetrized Laplacian provides a meaningful representation of the original directed Laplacian. 
To illustrate, Figure \ref{fig1} shows four simple directed graphs and the corresponding undirected graphs with equivalent effective resistance.

\begin{figure}[h!]
\centering
\includegraphics[width=3.4in]{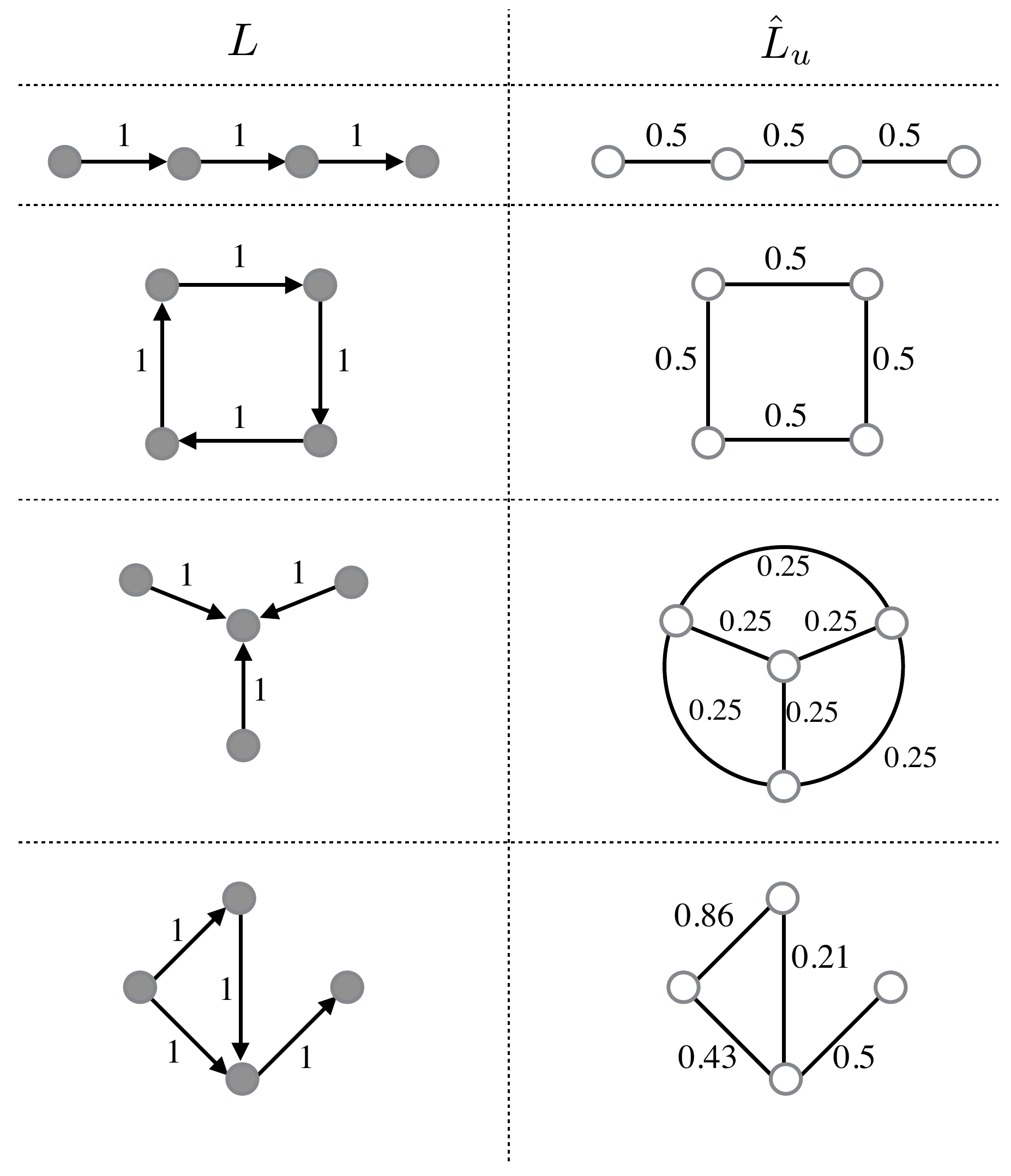}
\caption{Simple four-node directed graphs and the corresponding undirected graphs with equivalent effective resistances between all pairs of nodes. }
\label{fig1}
\end{figure}

In statistics, the inverse of an observed covariance matrix is known as a \emph{precision matrix}, which can be used to define an undirected Laplacian (often with a sparsification step)  \cite{Lauritzen1996}. The work here follows similar logic by connecting the steady-state covariance matrix from an directed graph to the undirected graph that generates the same steady-state covariance matrix. 
From a dynamics systems perspective this means that if each node in $L$ and $\hat{L}_u$ were subject to independent white noise of constant covariance, the expected steady-state covariance between any pair of nodes in $L$ is equivalent to the expected steady-state covariance between that pair of nodes in $\hat{L}_u$.  

\emph{Remark} 1. When studying graphs that admit negative edge weights, it is common to employ  the signed Laplacian. The signed Laplacian is defined similarly to the combinatorial Laplacian, with the adjustment that the degree of each node is the sum of the absolute values of all adjacent edges \cite{kunegis2010}. The primary motivations for this definition are to ensure that the signed Laplacian is at least positive semidefinite and that there are no zero entries along the main diagonal.  Drawbacks to working with the signed Laplacian are that row and column sums are no longer guaranteed to be zero, and interpretation of spectral properties is less straightforward. A Laplacian as defined by $\hat{L}_u=X^+$ is guaranteed to be positive semidefinite despite the presence of negative edge weights because $\Sigma$ is positive definite and the multiplications with $Q^T$ and $Q$ add an eigenvalue at zero.  Due to the guaranteed positive semidefiniteness of $\hat{L}_u$ the signed Laplacian is not used  in this work. 

\subsection{Exploring the relationship between $L$ and $\hat{L}_u$} \label{sec:explore}
The previous section demonstrates how one can use the inverse of the solution to the Lyapunov equation to derive an undirected graph with equivalent resistance distances. It is noted that though two nonisomorphic undirected graphs can have the same multiset of resistance distances, the resistance distance matrix $R$ with elements $r_{i,j}$ uniquely determines an undirected graph up to an isomorphism. This is a result of the one-to-one relationship between elements of $R$ and $L_u^+$ \cite{young20162} in combination with the uniqueness property of Moore--Penrose pseudoinverses. Subsequently, each directed graph $L$ has the same effective resistances as exactly one undirected graph $\hat{L}_u$. However, the reverse is not true and one undirected graph $\hat{L}_u$ can have equivalent effective resistances to many directed graphs. The following illustrates this property and describes the algebraic relationship required between $L$ and $\hat{L}_u$ such that the two graphs have equivalent resistance distances between nodes. 

It is shown above that $L$ and $\hat{L}_u$ have the same solution to the Lyapunov equation on $\mathbf{1}^{\perp}_n $. From a linear algebra perspective, in order for equations (\ref{lyapeff}) and (\ref{siglyap}) to hold, the following relation must be true (see \cite{Barnett1967}):
\begin{align}
\bar{L} &= \Big(\bar{K} + \frac{1}{2} I_{n-1}\Big) \Sigma^{-1}, \nonumber \\
\bar{L} &= \Big(\bar{K} + \frac{1}{2} I_{n-1}\Big) 2Q\hat{L}_uQ^T, \nonumber
\end{align}
where $\bar{K}$ is a skew symmetric $(n-1) \times (n-1)$ matrix satisfying
\begin{align}
\bar{L}\bar{K} + \bar{K}\bar{L}^T = \frac{1}{2}(\bar{L}-\bar{L}^T) \nonumber.
\end{align}
Pre- and postmultiplying by $Q^T$ and $Q$, respectively, and recalling that $Q^TQ = P_n$ and $P_n\hat{L}_u = \hat{L}_uP_n = \hat{L}_u$, yields
\begin{align}
Q^T \bar{L} Q = 2K\hat{L}_u + \hat{L}_u, \nonumber  
\end{align}
where $K = Q^T \bar{K} Q$. It follows from the construction that $K$ is a skew symmetric matrix with one eigenvalue at zero, and row and column sums are equal to the zero vector.  Furthermore, note that $Q^T \bar{L} Q = Q^TQLQ^TQ = P_nLP_n = P_nL$. Thus,
\begin{align}
P_nL = 2K\hat{L}_u + \hat{L}_u. \nonumber  
\end{align}
The matrix $P_n$ is an orthogonal projection matrix onto $\mathbf{1}^{\perp}_n $. Therefore, $P_nL$ has the property that the vector $\mathbf{1}_N$ is in the left and right null spaces. In other words, the row and column sums of $P_nL$ are equal to the zero vector. Premultiplying both sides by $H= L(P_nL)^+$ yields
\begin{align}
L = H(I_n + 2K )\hat{L}_u, \label{decomp}
\end{align}
where the properties from the definition of $H$ that $HP_n = H$ and $HL = L$ have been applied. The intuition behind the matrix $H$ is that it determines a node, or a set of nodes, that is globally reachable, as $\text{tr}(H) = n-1$ and the diagonal entries of $H$ that satisfy $H_{i,i}  <1$ give the set of nodes which can be set as a root for a spanning tree of $\mathcal{G}$. For example, if there is only one globally reachable node, $k$, and without loss of generality let $k$ be indexed as the $n$th node, then $H$ is of the form
\begin{align}
H = \left[ \begin{array}{ccc|c}
 \ddots & & & -1\\
 & I_{(n-1)\times (n-1)} & & \vdots\\
 & & \ddots & -1\\
\hline
0 & \cdots & 0 &0
\end{array} \right].
\end{align}
If the graph is an unweighted, directed cycle then $H=P$, that is, $H$ is the orthogonal projection matrix onto $\mathbf{1}^{\perp}_n $.

\section{Spectral Properties} \label{sec:spec}
As demonstrated in Section \ref{sec:explore}, one can always decompose a directed, connected Laplacian as $L = H(I_{N}+2K)\hat{L}_u$, where $\hat{L}_u$ is an undirected, connected Laplacian matrix defined on the same set of nodes where effective resistance between any two nodes, $i,j$ in $\hat{L}_u$ is equivalent to effective resistance between $i,j$ in $L$.  This section presents additional properties of $\hat{L}_u$ and demonstrates that the spectra of $\hat{L}_u$ characterizes features of $L$, thus further motivating the use of $\hat{L}_u$ as a meaningful symmetrization of $L$.

\subsection{Trace preservation, average cuts, and eigenvalue bounds}
 The following subsection introduces preserved quantities between $L$ and $\hat{L}_u$ and demonstrates that the eigenspectrum of $L$ is bounded by the eigenspectrum of $\hat{L}_u$.  
\begin{proposition}\label{propcuts}
Let $\mathcal{G} = (\mathcal{V},\mathcal{E},A)$ be a directed, connected graph of order $n$. Let  $\mathcal{\hat{G}}_u= (\mathcal{V},\mathcal{\hat{E}}_u,\hat{A}_u)$ be the associated undirected graph with equivalent effective resistance to $\mathcal{G}$. Then the mean value of all  graph cuts in $\mathcal{G}$ is equivalent to the mean value of all  cuts in  $\mathcal{\hat{G}}_u$.
\end{proposition}
\begin{proof}
It is clear from the decomposition $\bar{L} = (\bar{K} + \frac{1}{2} I_{n-1})\Sigma^{-1}$ that $\text{tr}(\bar{L}) = \frac{1}{2}\text{tr}(\Sigma^{-1})$. Applying (\ref{lbar}) then yields $\text{tr}({L}) = \text{tr}(\hat{L}_u)$. Since the trace of $L$ ($\hat{L}_u$) is equal to the negative sum of off-diagonal elements of $L$ ($\hat{L}_u$, respectively) it holds that the total sum of edge weights in $L$ is equal to the total sum of edge weights in $\hat{L}_u$.
Letting $k = 2^{n} - 2$ be the total number of possible bipartitions of the graph, the mean value of all graph cuts in $L$ and $\hat{L}_u$, respectively, can be written as
\begin{align} 
\bar{C}_L = \frac{1}{k} \text{tr}(Y (L+L^T) Y^T), \; \; \; \; \; \bar{C}_{\hat{L}_u} = \frac{1}{\frac{1}{2}k} \text{tr}(Y \hat{L}_u Y^T), \nonumber
\end{align}
where $Y$ is a $k \times n$ matrix where each row is a vector, $y$, corresponding to a partition as in (\ref{yvec}). Applying the property  that the trace of the product of three matrices is invariant under cyclic permutations gives
\begin{align}
\bar{C}_L &= \frac{1}{k} \text{tr}( (L+L^T) Y^TY)  \nonumber \\
&= \frac{1}{k} \text{tr}( (L+L^T) (2^{n-1}I + 2^{n-2} \mathbf{1}_n \mathbf{1}_n^T)) \nonumber \\
&= \frac{1}{k} \text{tr} (L+L^T) + \frac{1}{2k} \text{tr} (2^{n-2} \Delta_d^T \mathbf{1}_n^T) =  \frac{1}{k} \text{tr} (L+L^T) \nonumber \\
& = \frac{2}{k}\text{tr}(\hat{L}_u) = \frac{1}{\frac{1}{2}k} \text{tr}(Y \hat{L}_u Y^T) = \bar{C}_{\hat{L}_u}.
\end{align}
\end{proof}

In addition to the property that the sum of eigenvalues of $L$ is equal to the sum of eigenvalues of $\hat{L}_u$, it is shown in \cite{vogt1966} that the eigenvalues of $\Sigma^{-1}$ bound those of $\bar{L}$. Therefore,
\begin{align}
\lambda_2(\hat{L}_u) \leq \lambda_2(L), \;\;\;  \lambda_n(L) \leq \lambda_n(\hat{L}_u). \nonumber 
\end{align}

\subsection{Relationship between effective resistance and eigenvectors of $\hat{L}_u$} \label{sec:erev}
The eigenvectors of $\hat{L}_u$ are closely related to effective resistance. More specifically, the space of the eigenvalue-scaled eigenvectors of $\hat{L}_u$ is a Euclidean space that preserves effective resistance in $\hat{L}_u$ \cite{saerens2004} and therefore also in $L$. This can be observed by the following equation:
\begin{align}
r_{i,j} &= \Big(\mathbf{e}_n^{(i)} - \mathbf{e}_n^{(j)} \Big)^T X \Big(\mathbf{e}_n^{(i)} - \mathbf{e}_n^{(j)} \Big) \nonumber \\
&= \Big(\mathbf{e}_n^{(i)} - \mathbf{e}_n^{(j)} \Big)^T \hat{L}_u^+ \Big(\mathbf{e}_n^{(i)} - \mathbf{e}_n^{(j)} \Big) \nonumber \\
&= \Big(\mathbf{x}_i - \mathbf{x}_j \Big)^T U^T \hat{L}_u^+ U \Big(\mathbf{x}_i - \mathbf{x}_j \Big) \nonumber \\ 
&= \Big(\mathbf{x}_i - \mathbf{x}_j \Big)^T (\Lambda^{1/2})^T \Lambda^{1/2} \Big(\mathbf{x}_i - \mathbf{x}_j\Big) \nonumber \\
&= \Big(\mathbf{y}_i - \mathbf{y}_j \Big)^T \Big(\mathbf{y}_i - \mathbf{y}_j \Big), \label{res_decomp}
\end{align}
where $\mathbf{x}_i = U^T \mathbf{e}_i$ and ${\mathbf{y}}_i = \Lambda^{1/2}\mathbf{x}_i$. Thus the square root of effective resistance between two nodes is equivalent to the Euclidean distance between the corresponding elements in the transformed space.   Furthermore, as shown in \cite{fouss2005}, the space of eigenvectors corresponding to the $l$ largest eigenvalues of $\hat{L}_u^+$ ($l$ smallest nonzero eigenvalues of $\hat{L}_u$) approximately preserves effective resistance. Let $\mathbf{u}_k$ denote the eigenvector corresponding to the $k$th largest  eigenvector of $\hat{L}_u^+$. Let $\tilde{U} =[\mathbf{0}, \hdots , \mathbf{0}, \mathbf{u}_l, \mathbf{u}_{l+1}, \hdots, \mathbf{u}_n]$,  $\tilde{\Lambda} = \text{diag}[0, \hdots, 0, \lambda_l, \lambda_{l+1},\hdots,\lambda_{n}]$.  Then an approximation of effective resistance can be calculated as
\begin{align}
\tilde{r}_{i,j}= \Big(\tilde{\mathbf{y}}_i - \tilde{\mathbf{y}}_j \Big)^T \Big(\tilde{\mathbf{y}}_i - \tilde{\mathbf{y}}_j \Big),
\end{align}
where $\tilde{\mathbf{x}}_i =\tilde{U}\mathbf{e}_i$ and $\tilde{\mathbf{y}}_i = \tilde{\Lambda} \tilde{\mathbf{x}}_i.$ This approximation is bounded by the sum of the $l-1$ smallest eigenvalues of $\hat{L}_u^+$ \cite{fouss2005}:  
\begin{align}
\|r_{i,j} - \tilde{r}_{i,j} \| \leq \sum_{k = 1}^{l-1} \lambda_k(\hat{L}_u^+). \label{r_approx}
\end{align} 
The relationship between effective resistance and the eigenspectrum of $\hat{L}_u$ implies that the scaled eigenvectors of $\hat{L}_u$ contain information about a graph metric on $L$. This further motivates the use of $\hat{L}_u$ as a proxy for $L$ for spectral graph analysis applications.

\section{Application: Graph Bisection} \label{sec:bisec}
In this section, graph bisection is investigated as an application for effective resistance preserving graph symmetrization. First, undirected graph bisection and the Fiedler vector corresponding to the undirected graph Laplacian are reviewed.  Then, directed graph bisection based on the Fiedler vector of the symmetrized Laplacian $\hat{L}_u$ is introduced. Finally, bounds on the value of an undirected ratio cut in $\hat{L}_u$ in terms of cuts in $L$ are provided. 

\subsection{Undirected Graph Bisection}\label{sec:ugp}
In \cite{Fiedler1975}, Fiedler demonstrated that a valid bisection of an undirected graph can be obtained by applying the eigenvector corresponding to the first nonzero eigenvalue of an undirected Laplacian as an indicator vector. That is, partitioning such that all nodes corresponding to a negative eigenvector component nodes are in one set and all nodes with a positive component are in the complementary set. Criteria for connectedness of the resulting node sets were provided in \cite{Fiedler1975} and later generalized in \cite{urschel2014}. Because of this work, the eigenvector corresponding to the first nonzero eigenvalue of an undirected Laplacian is more commonly known as the Fiedler vector.  

From a graph cut perspective, the Fiedler vector can also be shown to be the solution to the relaxed version of the undirected ratio cut problem \cite{Hage1992} for which the number of clusters, $k$, is $k=2$.   Letting the value of an undirected cut be equal to the sum of the weights of edges crossed by the cut, the undirected ratio cut problem seeks to find the minimal cut(s) for which the number of nodes in each resulting subset is approximately equal. This problem is relevant to applications such as image segmentation \cite{Wang2003}. Mathematically, the value of an undirected ratio cut for $k=2$ is defined as
\begin{align}
\text{URC}(\mathcal{P}_u,\overline{\mathcal{P}}_u) = \frac{\text{cut}(\mathcal{P}_u,\overline{\mathcal{P}}_u)}{|\mathcal{P}_u | } + \frac{\text{cut}(\mathcal{P}_u,\overline{\mathcal{P}}_u)}{|\mathcal {\overline{P}}_u | }. \label{urc1}
\end{align}
By defining the vector $f$ relative to a subset $\mathcal{P}_u \in \mathcal{V}$ as
\begin{align}
f_i = \begin{cases}
  \sqrt{\frac{|\overline{\mathcal{P}}_u|}{ |\mathcal{P}_u|}} & \text{if } i \in \mathcal{P}_u, \\
  -\sqrt{\frac{|\mathcal{P}_u|}{ |\overline{\mathcal{P}}_u|}} & \text{if } i \in \overline{\mathcal{P}}_u,
\end{cases} \label{f}
\end{align}	
it can be shown that (see \cite{luxburg2006} for discussion)
\begin{align}
\text{URC}(\mathcal{P}_u,\overline{\mathcal{P}}_u) = \frac{f^T L_u f}{2n} .
\end{align}
This allows for the following discrete optimization problem for minimizing undirected ratio cut (\ref{urc1}):
\begin{align}
\min_{\mathcal{P}_u \subset \mathcal{V}} f^T L_u f \;\;\text{subject to}\; f \perp \mathbf{1}_n, \; f_i \text{ as defined in (\ref{f})},\; \| f \| = \sqrt{n}. \label{opt1}
\end{align}
The optimization problem (\ref{opt1}) is NP-hard \cite{wagner1993} and can be relaxed by allowing $f \in \mathbb{R}^n$, rather than discrete values. This results in 
\begin{align}
\min_{f \in \mathbb{R}^n} f^T L_u f \;\;\text{subject to}\; f \perp \mathbf{1}_n, \;  \| f \| = \sqrt{n}, \label{opt2}
\end{align}
which, by the Rayleigh--Ritz theorem, can be solved by letting $f$ be the Fiedler vector of $L_u$. To obtain a valid partition that serves as an approximate solution to (\ref{opt1}), the sign of $f$ is used as an indicator function where all nodes with a corresponding positive entry in $f$ are assigned to $\mathcal{P}_u$ and all nodes with a negative entry are assigned to $\mathcal{\overline{P}}_u$. 

\subsection{Directed Graph Bisection} \label{sec:dgp}
In this subsection, bisection of a directed graph is considered, and motivation is provided for applying the Fiedler vector of the associated symmetrized Laplacian $\hat{L}_u$  as an indicator function to partition the corresponding directed graph $L$. As demonstrated in Section \ref{sec:erev}, the eigenvectors of $\hat{L}_u$ are closely related to the effective resistance of $L$. Specifically, the absolute value of the difference between entries $i$ and $j$ in the Fiedler vector provides an approximation of the effective resistance between nodes $i$ and $j$. Consequently, applying the Fiedler vector of $\hat{L}_u$  as an indicator function to split $L$ yields a division where nodes have, on average,  smaller effective resistances to other nodes within the same partition than to nodes on the other side of the split. Therefore, nodes in either partition are close to one another with respect to a graph metric. 

This is significant because the graph metric, square root of effective resistance, implicitly accounts for directionality in the graph.  It follows that the resulting graph partition is also reflective of edge directionality. To illustrate, consider the four-node graph in Figure \ref{ex1}, where, from left to right, the images show the original directed graph $L$, the undirected graph with equivalent effective resistance $\hat{L}_u$, spectral partitioning by the Fiedler vector of $\hat{L}_u$, and the partitioning applied to $L$. 
\begin{figure}[t!]
\centering
\includegraphics[width=5in]{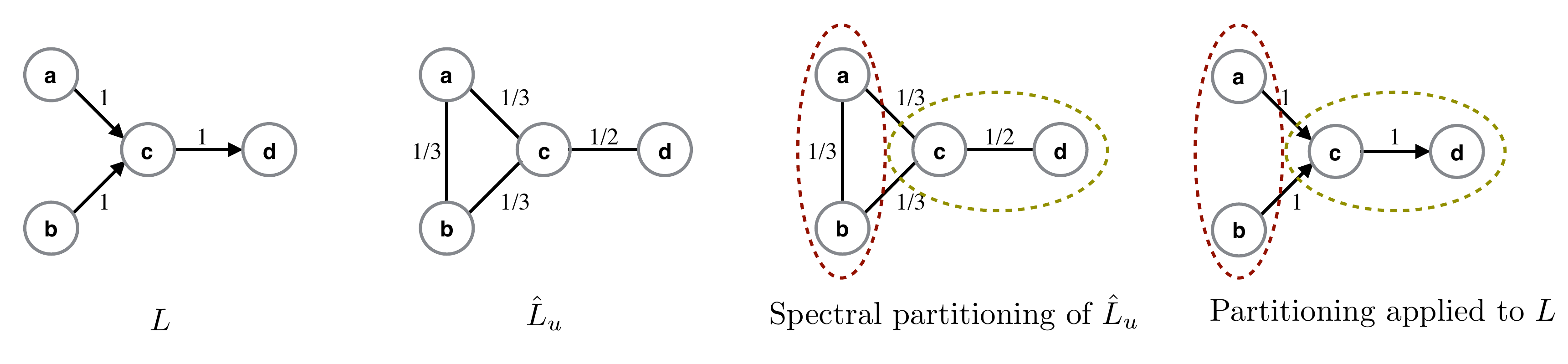}
\caption{Fiedler vector partitioning of simple directed graph with effective resistance preserving symmetrization.}
\label{ex1}
\end{figure}

In contrast, consider the undirected graph on the same set of nodes where edge directionality has been discarded, as is commonly done to bypass the difficulty of working with directed graphs. The result is a star graph where node $c$ is at the center and nodes $a$, $b$, and $d$ are identical in that their labels can be swapped without changing the graph structure. Due to this symmetry, the second smallest eigenvalue of the undirected graph Laplacian is repeated and the Fiedler vectors corresponding to these eigenvalues have a zero entry in the row corresponding to node $c$, which is treated as belonging to the positive side of the partition following the precident in \cite{Fiedler1975, urschel2014}. The Fiedler vector partitions are no longer unique and correspond to isolating a leaf node ($a$ $b$, or $d$), as shown in Figure $\ref{ex2}$.
\begin{figure}[t!]
\centering
\includegraphics[width=5in]{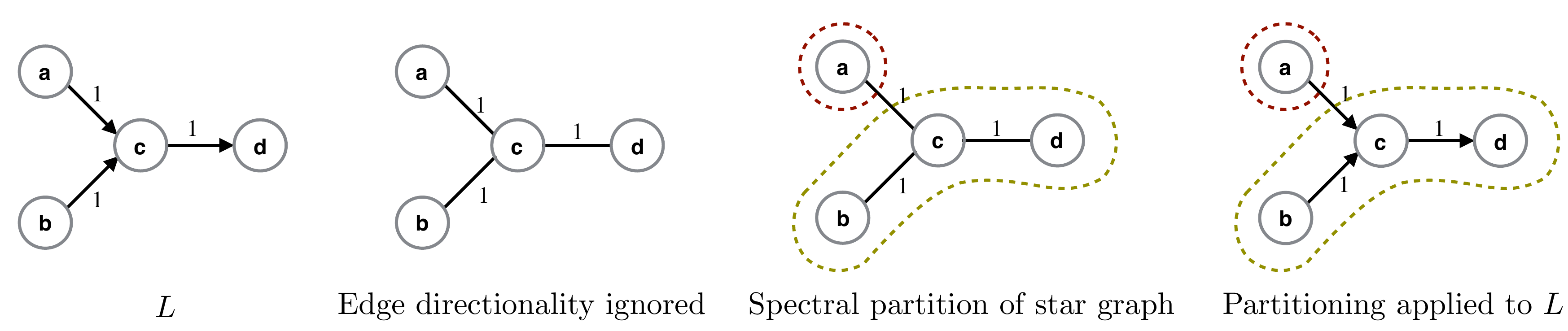}
\caption{Fiedler vector partitioning of simple directed graph where edge directionality has been ignored.}
\label{ex2}
\end{figure}

The partitioning of $L$ shown in Figure \ref{ex1}   is  more reflective of the underlying graph structure and yields partitions with more structural homogeneity \cite{leger2014} than the partitioning in Figure \ref{ex2}. Nodes $a$ and $b$ are similar in that they share a common destination, $c$; therefore it is expected that nodes $a$ and $b$ belong to the same partition. In Figure \ref{ex2}, the partition $\Mp = \{a\}$, $\Bmp = \{b,c,d\}$ is equally likely as $\Mp = \{d\}$, $\Bmp = \{a,b,c\}$. However, those partitions are structurally different in the context of the directed graph and would likely yield very different outcomes in graph analysis applications. 

As an additional example, consider the directed graph in Figure \ref{d_roach}, where the partitions obtained by the Fiedler vector of $\hat{L}_u$ are indicated by the colored ovals. It is easy to check that in this example, the Fiedler vector partitioning cuts the minimum number of edges needed to split the graph into two equal sized groups. 
\begin{figure}[t!]
\centering
\includegraphics[width=3in]{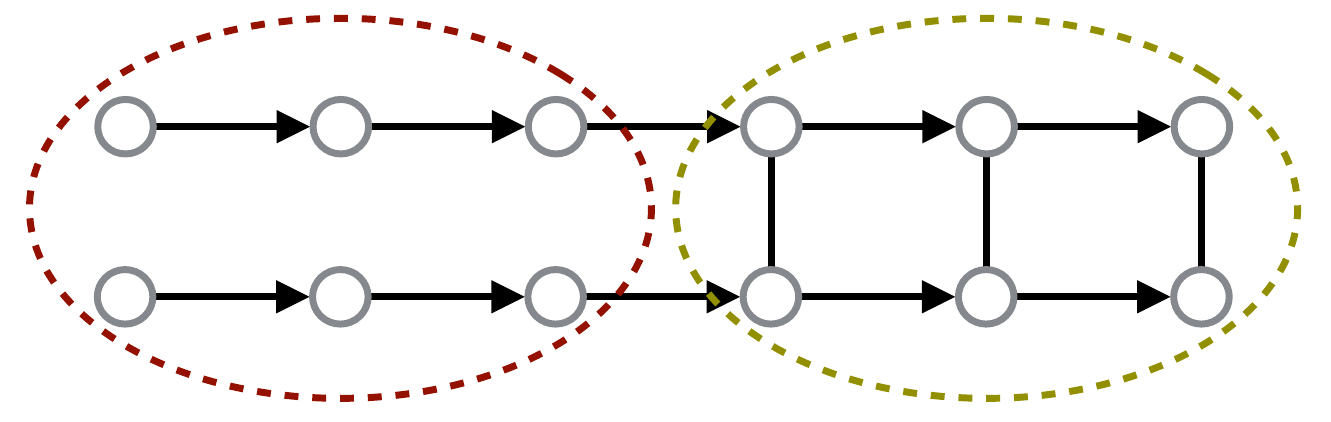}
\caption{Directed roach graph with partitions determined by the sign of the Fielder vector of $\hat{L}_u$.}
\label{d_roach}
\end{figure}

 Ignoring edge directionality yields the undirected roach graph \cite{guattery1998}, shown in Figure \ref{u_roach} with three vertical edges. For this family of graphs, it has been shown that Fiedler vector partitioning corresponds to cutting across the vertical edges connecting the upper and lower long paths, resulting in poor performance with respect to the value of the undirected ration cut as the graph size and number of vertical edges increases \cite{guattery1998}. 
\begin{figure}[t!]
\centering
\includegraphics[width=3in]{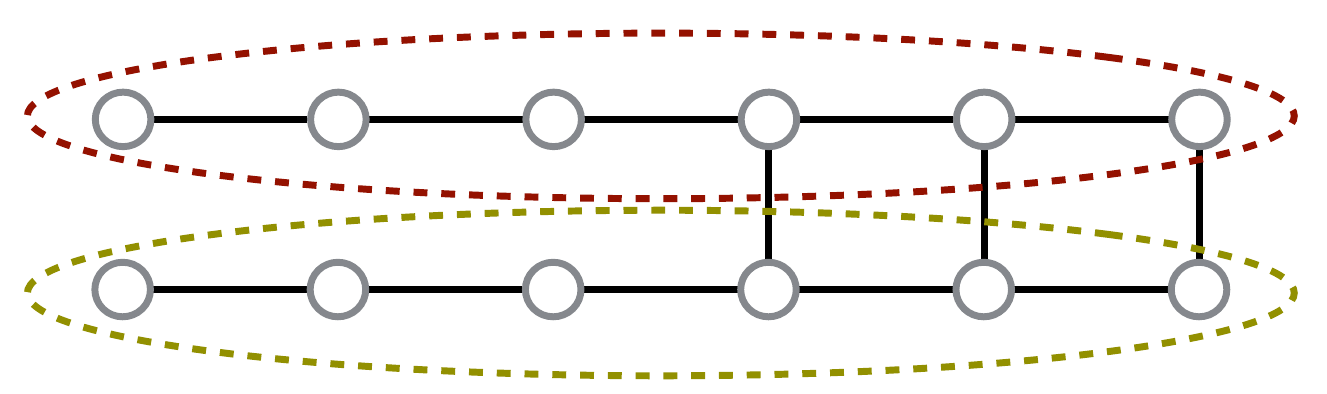}
\caption{Undirected roach graph with partitions determined by the sign of the Fielder vector of ${L}_u$.}
\label{u_roach}
\end{figure}

The comparison of Fiedler vector partitionings of the directed and undirected roach graphs is particularly interesting because the introduction of directionality improves the performance of spectral partitioning with respect to the number of cut edges. In general, this suggests that the performance of Fiedler vector partitioning of an undirected graph is not necessarily indicative of performance of Fiedler vector partitioning of a graph on the same set of nodes with some directed edges. Furthermore, the example reinforces that given a directed Laplacian, the symmetrized Laplacian $\hat{L}_u$ is structurally different from the Laplacian obtained by ignoring edge directionality. 

\subsection{Relating Undirected and Directed Ratio Cuts}
  As discussed in Section \ref{sec:ugp}, Fiedler vector partitioning of a symmetrized Laplacian $\hat{L}_u$ corresponds to approximating the minimum undirected ratio cut of $\hat{L}_u$. While Section \ref{sec:dgp}  demonstrated  that since $\hat{L}_u$ preserves a metric on $L$, the Fiedler vector partitioning of $\hat{L}_u$ yields a structurally meaningful partition when applied to $L$, it remains to be shown that there is a relationship between ratio cuts in $\hat{L}_u$ and in $L$.  To better understand what approximating the minimum of $f^T \hat{L}_uf$ means in terms of cuts in $L$, the directed ratio cut is first defined as
  \begin{align}
  \text{DRC}(\Mp,\Bmp) =  \frac{f^T Lf}{n} = \frac{\text{cut}(\mathcal{P},\overline{\mathcal{P}})}{|\mathcal{P} | } + \frac{\text{cut}(\Bmp,{\mathcal{P}})}{|\mathcal {\overline{P}} | }, \nonumber
  \end{align}
 where, in contrast to the undirected case, $\text{cut}(\mathcal{P},\overline{\mathcal{P}})$ is no longer by definition equal to $\text{cut}(\Bmp,{\mathcal{P}})$. 
 
 The term $f^T \hat{L}_uf$ can be rewritten in terms of $L$ as 
 \begin{align}
 f^T \hat{L}_uf = f^T (I_n + 2K)^{-1}P Lf = \frac{1}{2} f^T (I_n +S) PL f, \nonumber
 \end{align}
 where $S$ is an orthogonal matrix, that is, $S^T = S^{-1}$, defined by the Cayley transform $S = (I_n - 2K)(I_n + 2K)^{-1}$. Since $f \perp \mathbf{1}_n$ and vector multiplication with an orthogonal matrix preserves length, the following holds:
  \begin{align}
\frac{1}{2} f^T (I_n +S) PL f  = \frac{1}{2} f^T (I_n +S) L f = f^T L f + \frac{1}{2} f^T (S-I_n)L f. \label{qf1}
  \end{align}
The final equality of Equation (\ref{qf1}) is included to demonstrate that if $S = I_n$ (alternatively, $K = 0$), then $f^T \hat{L}_uf = f^TLf$ and $\text{URC}(\Mp,\Bmp) = \frac{1}{2} \text{DRC}(\Mp,\Bmp)$. Thus,
   \begin{align}
 f^T \hat{L}_uf  &=\frac{1}{2} f^T L f + \frac{1}{2} f^T SL f = \frac{1}{2} f^T L f + \frac{1}{2} \tilde{f}^T L f, \nonumber 
\end{align}
where
\begin{align}
\tilde{f}_k &=   \sqrt{\frac{|\overline{\mathcal{P}}|}{ |\mathcal{P}|}} \sum_{i \in \mathcal{P}} S_{i,k} - \sqrt{\frac{|\mathcal{P}|}{ |\overline{\mathcal{P}}|}} \sum_{j \in \mathcal{\overline{P}}} S_{j,k}. \nonumber 
\end{align}
The term $\sum_{i \in \mathcal{P}} S_{i,k}$ can be easily bounded as $- \sqrt{|{\mathcal{P}}|} \leq \sum_{i \in \mathcal{P}} S_{i,k} \leq \sqrt{| \mathcal{P}|}$ by applying the Cauchy--Schwarz inequality and the orthogonality of $S$. Similarly,  $- \sqrt{|\overline{\mathcal{P}}|} \leq \sum_{j \in \overline{\mathcal{P}}} S_{j,k} \leq \sqrt{|\overline{ \mathcal{P}}|}$. Therefore, 
\begin{align}
&  - \sqrt{|\overline{\mathcal{P}}|} -  \sqrt{|\mathcal{P}|}\leq \tilde{f}_k \leq \sqrt{|\overline{\mathcal{P}}|} + \sqrt{|\mathcal{P}|}.  \label{fkbound}
 \end{align}
   Furthermore, the elements of the vector $Lf$ can be expressed as
 \begin{align}
 Lf_i = \begin{cases}
 \frac{n}{\sqrt{|\mathcal{P}|| \overline{\mathcal{P}}|}} \sum_{k \in \overline{\mathcal{P}}} a_{i,k}  & \text{if } v_i \in \mathcal{P}, \\
  -\frac{n}{\sqrt{|\mathcal{P}|| \overline{\mathcal{P}}|}} \sum_{j \in {\mathcal{P}}} a_{i,j}  & \text{if } v_i \in\overline{\mathcal{P}}. 
\end{cases} \label{lfbound}
 \end{align}
 Consequently, upper and lower bounds on $f^TSLf$ can be obtained by combining Equations (\ref{fkbound}) and (\ref{lfbound}),
 \begin{align}
  f^TSLf &\geq -n \Big(\frac{1}{\sqrt{|\Mp|}} +\frac{1}{\sqrt{|\Bmp|}}\Big) \Big(\text{cut}(\mathcal{P},\overline{\mathcal{P}}) +  \text{cut}(\overline{\mathcal{P}},\mathcal{P}) \Big),  \nonumber \\
 f^TSLf &\leq n \Big(\frac{1}{\sqrt{|\Mp|}} +\frac{1}{\sqrt{|\Bmp|}} \Big) \Big(\text{cut}(\mathcal{P},\overline{\mathcal{P}}) +  \text{cut}(\overline{\mathcal{P}},\mathcal{P}) \Big). \nonumber
 \end{align}
 This leads to the following bounds on an undirected ratio cut in $\hat{L}_u$ in terms of cuts in $L$:
 \begin{align}
 \frac{f^T \hat{L}_uf}{2n} &\geq \frac{\text{cut}(\mathcal{P},\overline{\mathcal{P}})}{|\mathcal{P} | } + \frac{\text{cut}(\overline{\mathcal{P}},{\mathcal{P}})}{|\mathcal {\overline{P}} | }-\Big(\frac{1}{\sqrt{|\Mp|}} +\frac{1}{\sqrt{|\Bmp|}}\Big) \Big(\text{cut}(\mathcal{P},\overline{\mathcal{P}})+\text{cut}(\overline{\mathcal{P}},\mathcal{P})\Big), \nonumber \\
 \frac{f^T \hat{L}_uf}{2n}  &\leq \frac{\text{cut}(\mathcal{P},\overline{\mathcal{P}})}{|\mathcal{P} | } + \frac{\text{cut}(\overline{\mathcal{P}},{\mathcal{P}})}{|\mathcal {\overline{P}} | } + \Big(\frac{1}{\sqrt{|\Mp|}} +\frac{1}{\sqrt{|\Bmp|}}\Big) \Big(\text{cut}(\mathcal{P},\overline{\mathcal{P}})+\text{cut}(\overline{\mathcal{P}},\mathcal{P})\Big). \label{bounds}
 \end{align}

Though these bounds are loose, they establish that undirected ratio cuts in $\hat{L}_u$ are indeed bounded with respect to the value of directed ratio cuts in $L$. More specifically, Equation (\ref{bounds}) implies that letting the total cut, $\text{cut}_t(\mathcal{P},\overline{\mathcal{P}}) = \text{cut}(\mathcal{P},\overline{\mathcal{P}}) +  \text{cut}(\overline{\mathcal{P}},{\mathcal{P}} )$, and without loss of generality assuming $|\Mp| \leq |\Bmp|$,  
\begin{align}
 0 \leq \text{URC}(\Mp,\Bmp) \leq (1+2\sqrt{|\Mp|}) \frac{\text{cut}_t(\mathcal{P},\overline{\mathcal{P}})}{|\Mp|} .
\end{align}
Therefore, given a directed graph $L$ and its symmetrization $\hat{L}_u$,  any approximation or solution to the optimization problem (\ref{opt1}) for  $\hat{L}_u$ will give a partitioning in the associated directed graph $L$ that is at most $(1 + 2\sqrt{|\Mp|}) $ times the ratio between the total sum of directed cut edges and the cardinality of the smaller partition. The minimum of this ratio is often referred to as the edge expansion or sparsest cut of the graph \cite{arora2009}.


\section{Application: Node sparsification} \label{sec:spars}
Kron reduction is a node (equivalently, vertex) sparsification method originating from circuit theory where the Schur complement of a subset of circuit elements is used to define an electrically equivalent circuit~\cite{kron1965}. In \cite{dorfler2012}, the authors provide a graph-theoretic analysis of the Kron reduction process and demonstrate, among other results, that  effective resistance is preserved among nodes in a Kron-reduced graph. This property motivates an extension of Kron reduction to directed graphs, which will be discussed in the following subsection. 

The objective of extending Kron reduction to directed graphs is to construct a directed graph on a limited subset of nodes where the effective resistance between nodes in the reduced graph is equivalent to the effective resistance between those nodes in the original graph. Such a reduction is relevant to many large directed graph analysis problems where only a limited number of nodes are of interest. For example, consider a social network in which there is a small subset of users with a large number of followers. To better understand the influence of these ``important'' users on each other, one could construct a reduced graph with only those users. By preserving a metric, square root of effective resistance, on graphs it is ensured that if two users are at a large distance from one another in the original graph, the edge weighting in the reduced graph will be such that this distance remains the same. As a result, one can simulate interactions between these important users or how information diffuses between them without considering the entire, very large, social network. 

Before discussing the directed graph node sparsification,  Kron reduction in undirected graphs is first reviewed. Consider a subset of nodes $\mathcal{V}^\tkr \subset \mathcal{V}$ and the complementary subset $\mathcal{\overline{V}}^\tkr \subset \mathcal{V}$ in an undirected graph with Laplacian matrix $L_u$ partitioned as
\begin{align}
L_u = \left[ \begin{array}{c|c}
  L_{u_{\mathcal{V}^\tkr,\mathcal{V}^\tkr}} &  L_{u_{\mathcal{V}^\tkr,\mathcal{\overline{V}}^\tkr}}\\
\hline
L_{u_{\mathcal{\overline{V}}^\tkr,\mathcal{V}^\tkr}} & L_{u_{\mathcal{\overline{V}}^\tkr,\mathcal{\overline{V}}^\tkr}}
\end{array} \right].
\end{align}
Then the Kron-reduced Laplacian, $L_u^\tkr$, on the subset $\mathcal{V}^\tkr$ is defined by the Schur complement 
\begin{align}
L_u^\tkr \coloneqq L_{u_{\mathcal{V}^\tkr,\mathcal{V}^\tkr}} - L_{u_{\mathcal{V}^\tkr,\mathcal{\overline{V}}^\tkr}}L^{-1}_{u_{\mathcal{\overline{V}}^\tkr,\mathcal{\overline{V}}^\tkr}}L_{u_{\mathcal{\overline{V}}^\tkr,\mathcal{V}^\tkr}}. \label{krun}
\end{align} 
A thorough survey of properties of the  Kron-reduced, undirected Laplacian can be found in  \cite{dorfler2012}. In general, it is assumed that the subset of nodes that is of interest, $\mathcal{V}^\tkr $, has been determine a priori. For electrical applications, $\mathcal{V}^\tkr $ typically contains external nodes, or in other words the critical electric components on the network periphery. The subset could also be determine by taking only those with certain properties, such as high degree in the social network example. Alternatively, if one wants to cut the size of the network approximately in half, the eigenvector associated with the largest eigenvalue of $L_u$ could be applied as an indicator vector \cite{Shuman2016}. This gives a generalization of selecting every other node in the network. 

One straightforward approach for node sparsification of directed graphs is to replace $L_u$ with a directed graph Laplacian, $L$, in Equation (\ref{krun}). This can lead to computational difficulties when $L^{-1}_{\mathcal{\overline{V}}^\tkr,\mathcal{\overline{V}}^\tkr}$ is singular, which occurs when the set $\mathcal{\overline{V}}^\tkr$ contains a node with no outgoing edges. Additionally, the resulting reduced Laplacian could contain new self loops  and effective resistance between remaining nodes is typically no longer equivalent to their effective resistance in the original graph. As a result, there is in general no clear interpretation of the reduced graph Laplacian.  

Due to the fact that the effective resistance in $\hat{L}_u$ is equivalent to effective resistance in $L$, the aforementioned issues can be circumvented by calculating the Schur complement (\ref{krun}) using $\hat{L}_u$, such that an undirected graph on a subset of nodes, $\hat{L}_u^\tkr$, is obtained. Then, after computing $H^\tkr = H_{\mathcal{V}^\tkr,\mathcal{V}^\tkr}P_m$ and $K^\tkr = K_{\mathcal{V}^\tkr,\mathcal{V}^\tkr}P_m$, the reduced graph can be mapped from an undirected  to a directed graph by 
\begin{align}
L^\tkr = H^\tkr (I_m +2 K^\tkr) \hat{L}_u^\tkr. \nonumber
\end{align} A schematic drawing of this procedure is shown in Figure \ref{kr1}.
 
\begin{figure}[h!]
\centering
\includegraphics[width=2in]{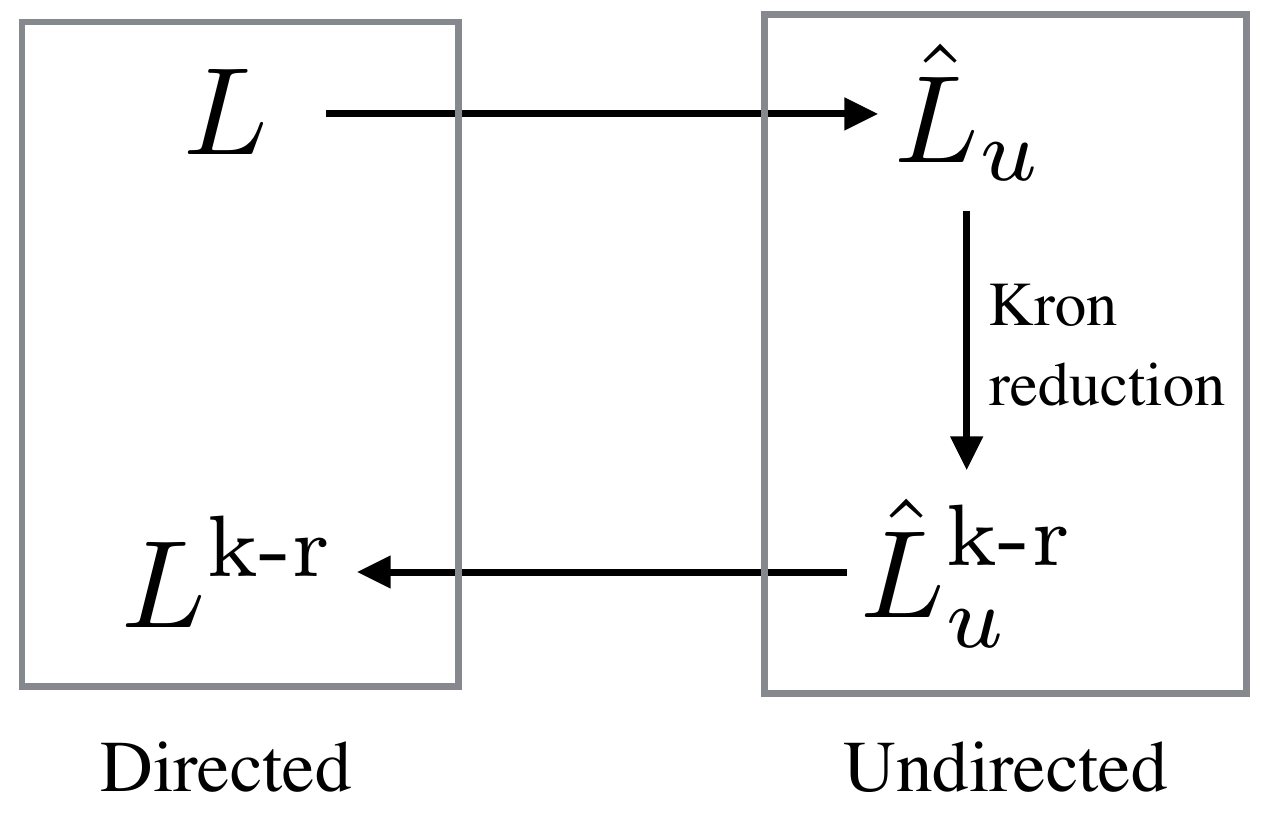}
\caption{A schematic drawing of the procedure for calculated a directed, Kron-reduced Laplacian}
\label{kr1}
\end{figure} 

It is noted that the resulting directed, reduced Laplacian is not unique since  $H^\tkr$ and $K^\tkr$ could be alternatively defined as any projection matrix and skew symmetric matrix, respectively, that meet the descriptions from Section \ref{sec:explore}. The projection of the submatrices of $H$ and $K$ associated with $\mathcal{V}^\tkr$ onto $\mathbf{1}_n^\perp$ provides an intuitive choice for the definitions of $H^\tkr$ and $K^\tkr$. Consider first the matrix $H$, which from a dynamical systems perspective represents damping or absorption of a dynamic process. If $1/n$ units of information (1 unit total over the graph) is given to each node and passed around the network through the directed edges, then the diagonal entries of the matrix $I_n - H$ represent how much information settles, or is absorbed, at each node when the dynamic process becomes stationary. By projecting the submatrix of $H$ associated with $\mathcal{V}^\tkr$ onto $\mathbf{1}_n^\perp$, it is ensured that the damping or absorption properties of a node are preserved in the reduced graph. A similar argument can be made for the skew symmetric matrix $K$, which accounts for rotation and subsequently directionality in the graph. The simple projection of the submatrix of $K$ associated with $\mathcal{V}^\tkr$ onto $\mathbf{1}_n^\perp$ preserves a notion of directionality of the connections between two nodes in the reduced graph. 

\section{Analysis of large graphs} \label{sec:large}
It has been demonstrated that in large undirected graphs ($n > 2000$), effective resistance between two nodes, $i$ and $j$, becomes proportional to the node degrees and structural information in $r_{i,j}$ is overshadowed by these terms. That is, $r_{i,j} \appropto 1/d_i + 1/d_j$, with a small remainder reflecting graph structure \cite{Luxburg2010}. Consequently, while all statements in this paper hold for any graph size, they are most meaningful with respect to the underlying structure when the graph is not too large. Nevertheless, if one has a distance-based problem such as clustering and chooses to symmetrize a large, directed graph using the approach from Section \ref{sec:sym}, it is important to analyze the resulting undirected graph with a meaningful distance measure. Two examples of distance measures applicable for the analysis of a large symmetrized Laplacian $\hat{L}_u$ are those induced by the amplified commute kernel \cite{Luxburg2010} and the heat kernel \cite{bai2004}. Recall that a distance measure $d_{\kappa_{i,j}}$, between two nodes $i$, $j$, induced by a kernel $\kappa$ is  $d_{\kappa_{i,j}} = \kappa_{i,i} + \kappa_{j,j} - 2 \kappa_{i,j}$.

A more rigorous and complete approach to studying large directed graphs would be to develop symmetrization methods that preserve distances that are strucurally representative even in large graphs. These include the distance measures induced by the aforementioned kernels as well as the class of graph-geodetic distances proposed in \cite{Chebotarev2011} and the communicability distance \cite{estrada2012}. However, these distances are not yet well-defined on directed graphs. Therefore, extending these distance measures to directed graphs and establishing existence of distance preserving symmetrized graphs are areas of important future research. 

\section{Final remarks} \label{sec:fr}
In this paper a new approach for directed graph symmetrization that preserves a graph metric is presented. It is shown that any connected, directed Laplacian matrix, $L$,  can be decomposed into the product of a projection matrix, a skew symmetric matrix, and the corresponding symmetrized Laplacian. The decomposition can be interpreted as the product of a damping (or absorption) matrix, a rotational matrix, and a symmetric positive semidefinite stability matrix.  The symmetrized Laplacian, $\hat{L}_u$ preserves a graph metric, the square root of effective resistance, of the original directed Laplacian. Many spectral properties of  $\hat{L}_u$ are shown to be reflective of properties of the $L$. For example, the rate of convergence to consensus for a directed graph with Laplacian $L$ in a simple dynamic process is dominated by the second smallest eigenvalue of $\hat{L}_u$. Additionally, the trace of $L$ and its symmetrization $\hat{L}_u$ are equivalent and the eigenvalues of the symmetrized Laplacian bound those of the directed Laplacian. These relationships stem from the fact that the eigenvalue-scaled eigenvectors of $\hat{L}_u$ form a Euclidean space that preserves effective resistance in $L$. 

The clear relationship between the spectrum of $\hat{L}_u$ and properties of $L$ motivate the application of $\hat{L}_u$ as a proxy for $L$ in graph analysis methods which require symmetry. One example is graph bisection, where it is shown that taking the sign of the Fiedler vector of $\hat{L}_u$ as an indicator vector yields a partition that is reflective of the structure of the directed graph $L$. Furthermore, the value of the cut in $\hat{L}_u$ generated by the partition is bounded relative to the cut generated by applying the same partition to $L$. A second example for the application of $\hat{L}_u$ is node sparsification, and it has been demonstrated that the Schur complement of $\hat{L}_u$ corresponding to a reduced subset of nodes can be used to define a graph on the reduced subset for which effective resistance is preserved.

The symmetrization method and directed Laplacian decomposition present a number of courses for future study. Given an undirected $\hat{L}_u$, one can generate a directed Laplacian, $L$, by choosing $H$ and $K$. However, the resulting $L$ is not necessarily insightful and could be dense. Therefore an interesting problem is to determine the matrix $K$ that results in a directed Laplacian that is as sparse as possible,  assuming the desired set of globally reachable nodes (and therefore $H$), and $\hat{L}_u$ are known. Additionally, there are many applications for which effective resistance preserving symmetrization could be applicable, such as edge sparsification and signal processing on directed graphs.

\FloatBarrier
\bibliographystyle{siamplain} 
\bibliography{graphsymbib} 

\end{document}